\documentclass[amsthm]{autart}
\usepackage{epsfig,graphicx,cite,makeidx}
\usepackage{amssymb,amsmath}
\usepackage{array}
\usepackage{url}

\newtheorem{assumption}[thm]{Assumption}

\makeatletter
\def\ExtendSymbol#1#2#3#4#5{\ext@arrow 0099{\arrowfill@#1#2#3}{#4}{#5}}
\def\RightExtendSymbol#1#2#3#4#5{\ext@arrow 0359{\arrowfill@#1#2#3}{#4}{#5}}
\def\LeftExtendSymbol#1#2#3#4#5{\ext@arrow 6095{\arrowfill@#1#2#3}{#4}{#5}}
\makeatother

\newcommand\myArrow[2][]{\ExtendSymbol{\leftharpoondown}{-}{\rightharpoonup}{#1}{#2}}
\newcommand\myRightarrow[2][]{\RightExtendSymbol{-}{-}{\longrightarrow}{#1}{#2}}
\newcommand{\norm}[1]{\left\lVert#1\right\rVert}
\newcommand{\set}[1]{\left\{#1\right\}}
\DeclareMathOperator{\rank}{rank} %
\DeclareMathOperator{\diag}{diag} %
\DeclareMathOperator{\E}{E} %

\begin{document}
\begin{frontmatter}
\title{Robust Decentralized Stabilization of Markovian
  Jump Large-Scale Systems: A Neighboring Mode Dependent Control Approach\thanksref{footnoteinfo}}

\thanks[footnoteinfo]{This work was supported by National Natural Science
  Foundation of China under Grant 61004044, Program for New Century Excellent
  Talents in University (11-0880), the Fundamental Research Funds for the
  Central Universities (WK2100100013), and the Australian Research
  Council. Corresponding author J.~Xiong. Tel. +0086-551-63607782.}
\author[China]{Shan Ma}\ead{shanma@mail.ustc.edu.cn},
\author[China]{Junlin Xiong}\ead{junlin.xiong@gmail.com},
\author[Australia]{Valery A. Ugrinovskii}\ead{v.ugrinovskii@gmail.com},
\author[Australia]{Ian R. Petersen}\ead{i.r.petersen@gmail.com}
\address[China]{Department of Automation, University of Science and Technology of China, Hefei 230026, China}
\address[Australia]{School of Engineering and Information Technology, University of New South Wales at the Australian Defence Force Academy, Canberra ACT 2600, Australia}


\begin{abstract}
This paper is concerned with the decentralized stabilization problem for a class of
uncertain large-scale systems with Markovian jump parameters. The controllers
use local subsystem states and neighboring mode information to generate local
control inputs. A sufficient condition involving rank constrained linear matrix
inequalities is proposed for the design of such controllers.
A numerical example is given to illustrate the developed theory.
\end{abstract}
\end{frontmatter}

\section{Introduction}
Many physical systems, such as power systems and economic systems, often
suffer from random changes in their parameters. These parameter changes may
result from abrupt environmental disturbances, component failures or repairs,
etc. In many cases, a Markov chain provides a suitable model to describe the
system parameter changes.  A Markovian jump system is a hybrid system with
different operation modes.  Each operation mode corresponds to a deterministic
system and the jumping transition from one mode to another is governed by a
Markov chain.  Recently, Markovian jump systems have received a lot of
attention and many control issues have been studied, such as stability and
stabilization~\cite{FLS10:tac,S06:tac}, time
delay~\cite{FGS09:auto,SLXZ07:auto}, filtering~\cite{SBA991:tac,WSGW08:auto},
$H_{2}$ control~\cite{DY08:auto}, $H_{\infty}$
control~\cite{LU07:tac,XC02:tac}, model reduction~\cite{ZHL03:scl}. For more
information on Markovian jump systems, we refer the reader
to~\cite{MY06:book}.

In this paper, we consider the decentralized stabilization problem for a class of uncertain
Markovian jump large-scale systems. The aim is to design a set of appropriate local feedback control
laws, such that the resulting closed-loop large-scale system
is stable even in the presence of uncertainties.
Recently, the decentralized stabilization problem for Markovian jump large-scale systems
has been investigated in the literature; see e.g.,~\cite{UP05:ijc, LUO07:auto} and the references therein. It is important to point out that the stabilizing techniques developed in~\cite{UP05:ijc, LUO07:auto} and many other papers are built upon an implicit assumption that the mode information
of the large-scale system must be known to all of the local controllers. In other words,
the mode information of all the subsystems must be measured and then broadcast
to every local controller. Such an assumption,
however, may be unrealistic either because the broadcast of mode information
among the subsystems is impossible in practice or because the implementation is expensive.

To eliminate the need for broadcasting mode information,
a local mode dependent control approach has been developed in~\cite{XUP09:tac,XUP10:book}.
This control approach is fully decentralized. The local controllers
use only local subsystem states or outputs and local subsystem mode information
to generate local control inputs.
To emphasize this feature, this type of controller is referred to
as a \emph{local mode dependent controller} in~\cite{XUP09:tac,XUP10:book}.
As pointed out in~\cite{XUP09:tac,XUP10:book},
the local mode dependent control approach offers many advantages in practice.
First, it eliminates the need for broadcasting mode information among the subsystems
and hence is more suitable for practical applications.
Second, it significantly reduces the number of control gains
and hence results in cost reduction, easier installation and maintenance.

In this paper, we focus on the state feedback case of Markovian jump
large-scale systems and aim to build a bridge between the results
in~\cite{UP05:ijc} and~\cite{XUP09:tac}.  We assume that each local controller
is able to access and utilize mode information of its neighboring subsystems
including the subsystem it controls.  This assumption is motivated by the fact
that some subsystems may be close to each other in practice and hence exchange
of mode information may be possible among these subsystems.  Under this
assumption, we develop an approach, which we call a~\emph{neighboring mode
  dependent control approach}, to stabilize Markovian jump large-scale
systems.  Compared to the local mode dependent control approach, our approach
can stabilize a wider range of large-scale systems in practice.  It is
demonstrated in the numerical section that the system performance will improve
as more detailed mode information is available to the local controllers.
Hence the system performance achieved by our approach is better than that
achieved by the local mode dependent control approach. Furthermore, both the
global and the local mode dependent control approaches proposed
in~\cite{UP05:ijc} and~\cite{XUP09:tac} can be regarded as special cases of
the neighboring mode dependent control approach.

\textit{Notation:} $\mathbb{R}^{+}$ denotes the set of positive real numbers;
$\mathbb{S}^{+}$ denotes the set of positive definite  matrices;
$\mathbb{R}^{m}$ denotes the set of real $m \times 1$ vectors;
$\mathbb{R}^{m \times n}$ denotes the set of real $m \times n$ matrices.
$\diag[F_{1},\ldots,F_{N}]$ denotes a block diagonal matrix with $F_{1},\ldots,F_{N}$
on the main diagonal.
$I$ is the identity matrix.
For real symmetric matrices $X$ and $Y$,
$X\ge Y$ (respectively, $X>Y$) means that $X-Y$ is positive semi-definite
(respectively, positive definite).
$\norm{\cdot}$ denotes either the Euclidean norm for vectors or the
induced $2$-norm for matrices. The superscript ``$T$'' denotes transpose of a vector  or a matrix.
$\E(\cdot)$ denotes the expectation operator with respect to the underlying complete
probability space $(\Omega,\mathcal{F},\Pr)$.

\section{Problem Formulation}
\label{sec:pf}
Consider a Markovian jump large-scale system $\mathcal{S}$ comprising $N$ subsystems
$\mathcal{S}_{i}$, $i\in\mathcal{N}\triangleq\{1,2,\ldots,N\}$.
The $i$th subsystem $\mathcal{S}_{i}$ is of the following form~\cite{XUP09:tac}:
\begin{equation}\label{sys}
  \mathcal{S}_{i}:
  \left\{\begin{aligned}
    \dot{x}_{i}(t) &= A_{i}(\eta_{i}(t))x_{i}(t) + B_{i}(\eta_{i}(t))u_{i}(t) \\
    &\quad + E_{i}(\eta_{i}(t))\xi_{i}(t) + L_{i}(\eta_{i}(t))r_{i}(t), \\
    \zeta_{i}(t) &= H_{i}(\eta_{i}(t))x_{i}(t),
  \end{aligned}\right.
\end{equation}
where $x_{i}(t)\in\mathbb{R}^{n_{i}}$ is the state, $u_{i}(t)\in\mathbb{R}^{m_{i}}$ is the input,
$\xi_{i}(t)\in\mathbb{R}^{g_{i}}$ is the local uncertainty input,
$r_{i}(t)\in\mathbb{R}^{s_{i}}$ is the interconnection input, which describes
the effect of the other subsystems $\mathcal{S}_{j}$, $j\ne i$, on
$\mathcal{S}_{i}$. $\zeta_{i}(t)\in\mathbb{R}^{h_{i}}$ is the uncertainty output.
The initial state $x_{i}(0)$ is denoted by $x_{i0}$.
The random process $\eta_{i}(t)$ denotes the mode switching of the subsystem $\mathcal{S}_{i}$;
it takes values in a finite set $\mathcal{M}_{i}\triangleq\set{1,2,\ldots,M_{i}}$. The structure of $\mathcal{S}_{i}$ is shown in~Fig.~\ref{system}.

\begin{figure}[htbp]
\begin{center}
\includegraphics[width=8cm]{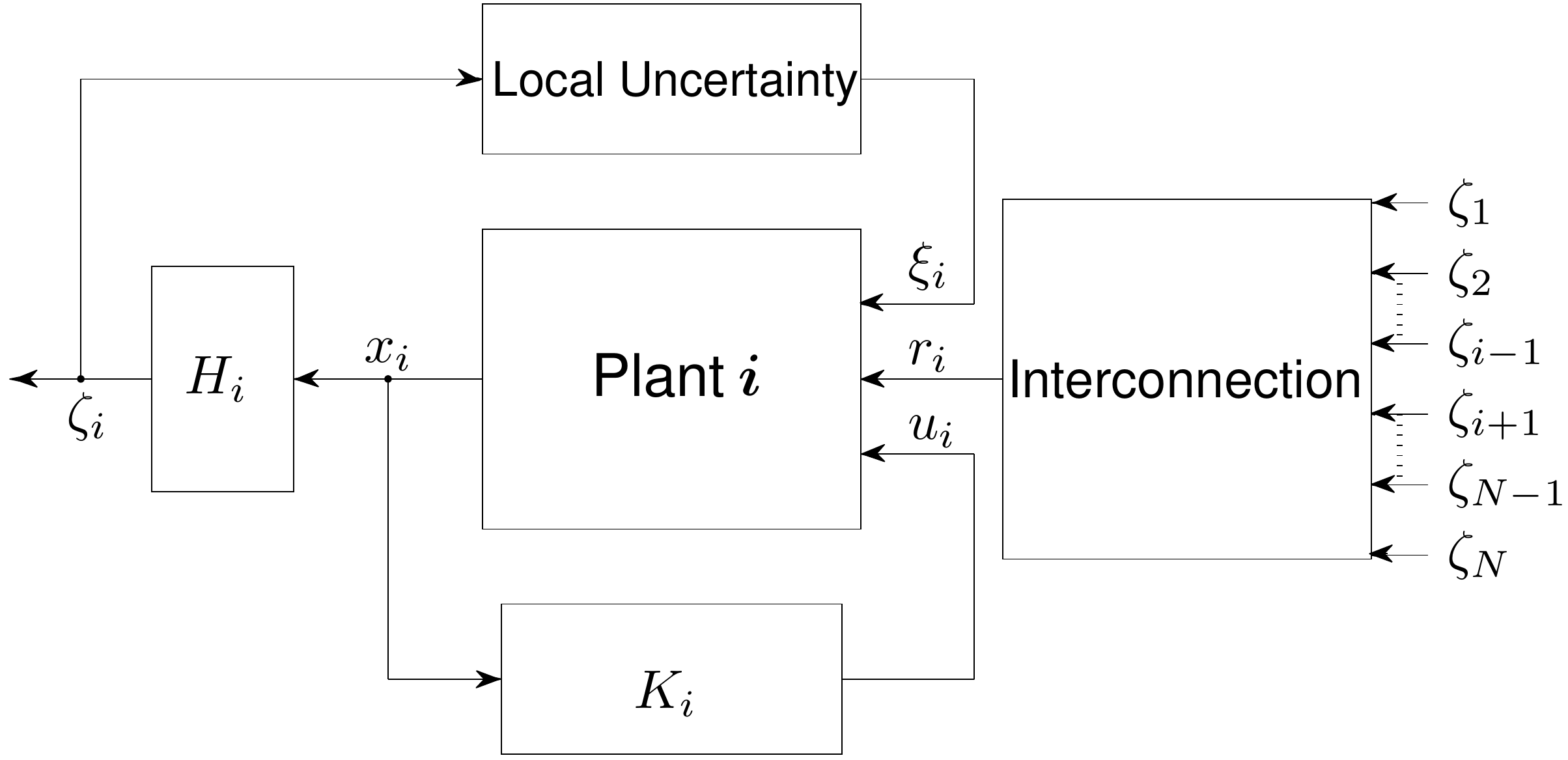}
\caption{The subsystem $\mathcal{S}_{i}$.}
\label{system}
\end{center}
\end{figure}

Because the mode process $\eta_{i}(t)$ describes the mode switching of the $i$th subsystem $\mathcal{S}_{i}$,
the vector process $[\eta_1(t),\ldots,\eta_N(t)]^{T}$ naturally describes the mode switching of
the entire large-scale system $\mathcal{S}$. We assume that $[\eta_1(t),\ldots,\eta_N(t)]^{T}$ takes values
in a set (denoted by $\mathcal{M}_{V}$) consisting of $M$ distinct vectors.
If any $\eta_{i}(t)$, $i\in\mathcal{N}$, changes its value,
the vector $[\eta_1(t),\ldots,\eta_N(t)]^{T}$ will take a different value.
Hence $M_{i}\le M$.
In addition, we have $M\le\prod_{i=1}^{N}M_{i}$ (not necessarily $``="$, because the
 mode processes $\eta_{i}(t)$, $i\in\mathcal{N}$, may depend on each
other~\cite{XUP09:tac}). Let $\mathcal{M}_{S}\triangleq\{1,2,\ldots,M\}$,
then a bijective mapping $\psi:\mathcal{M}_{V}\to\mathcal{M}_{S}$ exists, because $\mathcal{M}_{V}$ and $\mathcal{M}_{S}$ have the same number of elements.
Let $\eta(t)\triangleq\psi([\eta_1(t),\ldots,\eta_N(t)]^{T})$.
Thus the random vector process $[\eta_1(t),\ldots,\eta_N(t)]^{T}$ is transformed
into the random scalar process $\eta(t)$, which carries the same mode information
of the large-scale system $\mathcal{S}$. For this reason, $\eta(t)$ is referred to as the global mode process
in the sequel. The inverse function $\psi^{-1}:\mathcal{M}_{S}\to\mathcal{M}_{V}$
is given by $\psi^{-1}(\mu)=[\mu_{1},\ldots,\mu_{N}]^{T}$, $\mu\in\mathcal{M}_{S}$, $\mu_{i}\in\mathcal{M}_{i}$, $i\in\mathcal{N}$.
Then the $i$th element $\mu_{i}$ can be determined uniquely from the global mode $\mu$.
Therefore $\mu_{i}$ is also a function of $\mu$ and we write: $\mu_{i} = \psi_{i}^{-1}(\mu)$,
$i\in\mathcal{N}$.

We assume here
that $\eta(t)$ is a stationary Markov process.
The infinitesimal generator matrix of $\eta(t)$ is  $\mathbf{Q}=[q_{\mu\nu}]\in\mathbb{R}^{M\times M}$,
where $q_{\mu\nu}\ge0$ if $\nu\ne\mu$, and $q_{\mu\mu}=-\sum_{\nu=1,\nu\ne\mu}^{M}q_{\mu\nu}$.
The initial distribution of the process $\eta(t)$ is $\pi=[\pi_{1},\ldots,\pi_{M}]^{T}$ with $\pi_{\mu}\geq 0$, $\forall\mu\in\mathcal{M}_{S}$.

\begin{assumption}[\cite{UP05:ijc}]
 \label{locally:squire:integrable}
  Given any locally square integrable signals $u_{i}(t)$, $\xi_{i}(t)$, $r_{i}(t)$,
  for any initial conditions $x_{i}(0)=x_{i0}$, $\eta_{i}(0)=\eta_{i0}$,
  the solution $x_{i}(t)$ to each subsystem~\eqref{sys}
  exists and is locally square integrable.
\end{assumption}

\begin{rem}
  Recall that a signal $s(t)$ is said to be locally square integrable if it
  satisfies the condition $\E\left(\int_{0}^{\mathcal{T}}\norm{s(t)}^{2}dt\right)<\infty$
  for any finite time $\mathcal{T}$.
The term ``locally'' here means that square integrability is only required on bounded time intervals.
\end{rem}

The local uncertainty inputs and the interconnection inputs of the
large-scale system~\eqref{sys} are
assumed to satisfy the following integral quadratic
constraints (IQCs).

\begin{defn}[\cite{LUO07:auto}]
  \label{defn:un:local}
  Given a set of positive definite matrices $\bar{S}_{i}$, $i\in\mathcal{N}$.
  A locally square integrable signal $[\xi_{1}^{T}(t),\ldots,\xi_{N}^{T}(t)]^{T}$
  represents an
  admissible local uncertainty input for the large-scale system~\eqref{sys}
  if, given any locally square
  integrable signals $[u_{1}^{T}(t),\ldots,u_{N}^{T}(t)]^{T}$, $[r_{1}^{T}(t),\ldots,r_{N}^{T}(t)]^{T}$, there exists a time sequence $\set{t_{l}}_{l=1}^{\infty}$, $t_{l}\to\infty$,
  such that for all $l$ and for all $i\in\mathcal{N}$,
  \begin{align}
    \E\left(\int_{0}^{t_{l}}\left[\norm{\zeta_{i}(t)}^{2}
     -\norm{\xi_{i}(t)}^{2}\right] dt\right)
    \ge -x_{i0}^{T}\bar{S}_{i}x_{i0}.     \label{un:local}
  \end{align}
  The set of all such admissible local uncertainty inputs is denoted by $\Xi$.
\end{defn}

\begin{defn}[\cite{LUO07:auto}]
  \label{defn:un:ic}
  Given a set of positive definite matrices $\tilde{S}_{i}$, $i\in\mathcal{N}$.
  A locally square integrable signal $[r_{1}^{T}(t),\ldots,r_{N}^{T}(t)]^{T}$
  represents an admissible interconnection input for the large-scale system~\eqref{sys}
  if,
  given any
  locally square integrable signals $[u_{1}^{T}(t),\ldots,u_{N}^{T}(t)]^{T}$,
   $[\xi_{1}^{T}(t),\ldots,\xi_{N}^{T}(t)]^{T}$, there exists a time sequence $\set{t_{l}}_{l=1}^{\infty}$,
   $t_{l}\to\infty$, such that for all $l$ and for all $i\in\mathcal{N}$,
  \begin{align}
    &\E\left(\int_{0}^{t_{l}}\left[\sum\limits_{j=1,j\ne i}^{N}\norm{\zeta_{j}(t)}^{2}-
     \norm{r_{i}(t)}^{2}\right]
     dt\right)\ge -x_{i0}^{T}\tilde{S}_{i}x_{i0}. \label{un:ic}
  \end{align}
The set of all such admissible interconnections is denoted by $\Pi$.
We assume that the same time sequences
$\{t_{l}\}_{l=1}^{\infty}$  are chosen in Definition~\ref{defn:un:local} and Definition~\ref{defn:un:ic} whenever they correspond to the
same signals $[\xi_{1}^{T}(t),\ldots,\xi_{N}^{T}(t)]^{T}$, $[r_{1}^{T}(t),\ldots,r_{N}^{T}(t)]^{T}$, $[u_{1}^{T}(t),\ldots,u_{N}^{T}(t)]^{T}$.
\end{defn}

\begin{rem}
  The IQCs are used to describe relations between the input and output signals
  in the uncertainty blocks in~Fig.~\ref{system}. The constant terms on the
  right-hand sides of the inequalities~\eqref{un:local} and~\eqref{un:ic}
  allow for nonzero initial conditions in the uncertainty dynamics.  These
  definitions can capture a broad class of uncertainties such as nonlinear,
  time-varying, dynamic uncertainties; see~\cite[Chapter 2.3]{PUS00:book} for
  details.
\end{rem}

Let $\mathcal{C}=[c_{ij}]\in \mathbb{R}^{N\times N}$ be a given binary matrix,
where $c_{ij}=1$ if the mode of the subsystem $\mathcal{S}_{j}$ is available
to the $i$th local controller and $c_{ij}=0$, otherwise. Then the total mode
information accessed by the $i$th local controller can be written as
$[c_{i1}\eta_1(t),\ldots,c_{iN}\eta_N(t)]^{T}$. A zero entry in this vector
means that the mode information of the corresponding subsystem is not
available. We assume that the random vector process
$[c_{i1}\eta_1(t),\ldots,c_{iN}\eta_N(t)]^{T}$ takes values in a set (denoted
by $\mathcal{M}_{V i}$) consisting of $M_{ci}$ distinct vectors.
Obviously, $M_{ci}\le M$, $i\in\mathcal{N}$.  Let
$\mathcal{M}_{S i}\triangleq \{ 1,\ldots,M_{ci}\}$.  Also, there exists a
bijective mapping $\varphi_{i}:\mathcal{M}_{V i}\to\mathcal{M}_{S i}$
with $\varphi_{i}([c_{i1}\mu_1,\ldots,c_{iN}\mu_N]^{T})=\sigma_{i}$,
$\mu_{i}\in\mathcal{M}_{i}$, $\sigma_{i}\in\mathcal{M}_{S i}$,
$i\in\mathcal{N}$. Let
$\aleph_{i}(t)\triangleq\varphi_{i}([c_{i1}\eta_1(t),\ldots,c_{iN}\eta_N(t)]^{T})$,
$i\in\mathcal{N}$.  It can be seen that $\aleph_{i}(t)$ contains essentially
the same mode information as
$[c_{i1}\eta_1(t),\ldots,c_{iN}\eta_N(t)]^{T}$. Hence $\aleph_{i}(t)$ is
referred to as a neighboring mode process in the sequel.

\begin{rem}
Both the global and the local mode dependent control problems
 studied in~\cite{UP05:ijc,XUP09:tac} can be regarded as special cases
 of the neighboring mode dependent control problem with
 $\mathcal{C}=\mathbf{1}_{N\times N}$ (a matrix with all the elements being ones)
 and $I$, respectively.
\end{rem}

For the large-scale system~\eqref{sys} with the uncertainty constraints~\eqref{un:local},~\eqref{un:ic},
our objective is to design a
neighboring mode dependent decentralized control law
\begin{align}
  u_{i}(t) = K_{i}(\aleph_{i}(t))x_{i}(t),\quad i\in\mathcal{N}, \label{ctr}
\end{align}
such that the resulting closed-loop large-scale system is robustly stochastically stable
in the following sense.

\begin{defn}[\cite{XUP09:tac}]
 The closed-loop large-scale system corresponding to the
 uncertain large-scale system~\eqref{sys},~\eqref{un:local},~\eqref{un:ic} and the controller~\eqref{ctr} is said to be robustly
  stochastically stable if there exists a finite constant
  $\lambda\in\mathbb{R}^{+}$ such that
  \begin{align}
    \E \left(\int_{0}^{\infty}\sum\limits_{i=1}^{N}\norm{x_{i}(t)}^{2} dt\right)\le \lambda\norm{x_{0}}^{2} \label{eq:as}
  \end{align}
  for any $x_{0}=[x_{10}^{T},\ldots,x_{N0}^{T}]^{T}$, and any
 uncertainties $[\xi_{1}^{T}(t),\ldots,\xi_{N}^{T}(t)]^{T}\in\Xi$, $[r_{1}^{T}(t),\ldots,r_{N}^{T}(t)]^{T}\in\Pi$.
\end{defn}

For convenience, a set of many-to-one mappings $\phi_{i}$:
$\mathcal{M}_{S}\to\mathcal{M}_{S i}$, $i\in\mathcal{N}$, is introduced below:
  \begin{align}
\phi_{i}(\mu)&=\varphi_i(\diag[c_{i1},\ldots,c_{iN}]\cdot\psi^{-1}(\mu)).
\label{function}
 \end{align}
Note that $\phi_{i}$, $i\in\mathcal{N}$,
are also surjective mappings.

\begin{exmp}
  \label{exmp:1}
  Suppose $N=3$, $M_{1}=M_{2}=M_{3}=2$. When the mode processes $\eta_{i}(t)$, $i\in\mathcal{N}$,
  are independent of each other, the vector set $\mathcal{M}_{V}$ contains $8$
  elements, i.e., $\mathcal{M}_{V}=\{[\mu_{1},\mu_{2},\mu_{3}]^{T}: \mu_{i}=1, 2, i=1, 2, 3\}$.
  Now we assume that the mode processes $\eta_{i}(t)$, $i\in\mathcal{N}$,
  are subject to the constraints below:
  \begin{align*}
   \eta_1(t)=\eta_2(t) \quad &\text{if}\quad  \eta_3(t)=1, \\
     \eta_2(t)=1  \quad\quad\;\; &\text{if}\quad  \eta_1(t)=2.
  \end{align*}
  Then $\mathcal{M}_{V}$ contains only four
  elements, i.e., $\mathcal{M}_{V}=\{[1, 1, 1]^{T}, [1, 1, 2]^{T}, [1, 2, 2]^{T}, [2, 1, 2]^{T}\}$.
  Thus $\mathcal{M}_{S}=\{1, 2, 3, 4\}$.
  The mappings $\psi$, $\psi^{-1}$ between $\mathcal{M}_{V}$ and $\mathcal{M}_{S}$
  can be defined as follows:
  \begin{align*}
  \begin{matrix}
  &[1, 1, 1]^{T} \myArrow[\psi^{-1}]{\psi}  1,\quad
  [1, 1, 2]^{T} \myArrow[\psi^{-1}]{\psi}  2,\\
  &[1, 2, 2]^{T} \myArrow[\psi^{-1}]{\psi}  3, \quad
  [2, 1, 2]^{T} \myArrow[\psi^{-1}]{\psi}  4.
  \end{matrix}
  \end{align*}

  Suppose, for example, that
  $\mathcal{C} =\begin{bmatrix} 1 &  1   &  0 \\ 0 & 1  &  0 \\ 0 &  1 & 1 \\ \end{bmatrix}$.
  Then we have
  $\mathcal{M}_{V 1}=\{[1, 1, 0]^{T}, [1, 2, 0]^{T}, [2, 1, 0]^{T}\}$.
  Thus $\mathcal{M}_{S 1}=\{1, 2, 3\}$.
  The mapping $\varphi_{1}:\mathcal{M}_{V 1}\to\mathcal{M}_{S 1}$ can be defined as follows:
  \begin{align*}
  \begin{matrix}
  [1, 1, 0]^{T} \myRightarrow[\rule{0.5cm}{0cm}]{\varphi_{1}} 1,\quad
  [1, 2, 0]^{T} \myRightarrow[\rule{0.5cm}{0cm}]{\varphi_{1}} 2, \quad
  [2, 1, 0]^{T} \myRightarrow[\rule{0.5cm}{0cm}]{\varphi_{1}} 3.
  \end{matrix}
  \end{align*}
  In this case, by~\eqref{function}, the many-to-one mapping
  $\phi_{1}:\mathcal{M}_{S}\to\mathcal{M}_{S 1}$ is given by:
  \begin{align*}
  \begin{matrix}
  1 \myRightarrow[\rule{0.5cm}{0cm}]{\phi_{1}} 1,\quad
  2 \myRightarrow[\rule{0.5cm}{0cm}]{\phi_{1}} 1,\quad
  3  \myRightarrow[\rule{0.5cm}{0cm}]{\phi_{1}} 2, \quad
  4  \myRightarrow[\rule{0.5cm}{0cm}]{\phi_{1}}  3.
  \end{matrix}
  \end{align*}
  \end{exmp}

\section{Controller Design}
\label{sec:cd}
In this section, we first turn to a new uncertain Markovian jump
large-scale system which is similar to the large-scale system~\eqref{sys}.
Global mode dependent stabilizing controllers are designed for
this new large-scale system using the results of~\cite{UP05:ijc}.
Then we will show how to derive neighboring mode dependent stabilizing controllers
for the large-scale system~\eqref{sys} from these obtained global mode dependent controllers.
Finally, all of the conditions for the existence of
such neighboring mode dependent controllers
are combined as a feasible LMI problem with rank constraints.

Consider a new large-scale system $\tilde{\mathcal{S}}$ comprising $N$ subsystems $\tilde{\mathcal{S}_{i}}$, $i\in\mathcal{N}$.
The $i$th subsystem $\tilde{\mathcal{S}_{i}}$ is as follows~\cite{XUP09:tac}:
\begin{align}
\label{sys:large}
  \tilde{\mathcal{S}_{i}}:
  \left\{\begin{aligned}
    \dot{\tilde{x}}_{i}(t) &= \tilde{A}_{i}(\eta(t))\tilde{x}_{i}(t)
                    + \tilde{B}_{i}(\eta(t))\left[\tilde{u}_{i}(t)+\tilde{\xi}^{u}_{i}(t)\right] \\
    &\qquad + \tilde{E}_{i}(\eta(t))\tilde{\xi}_{i}(t) + \tilde{L}_{i}(\eta(t))\tilde{r}_{i}(t), \\
    \tilde{\zeta}_{i}(t)   &= \tilde{H}_{i}(\eta(t))\tilde{x}_{i}(t),
  \end{aligned}\right.
\end{align}
where $\tilde{A}_{i}(\mu)=A_{i}(\mu_{i})$,
$\tilde{B}_{i}(\mu)=B_{i}(\mu_{i})$, $\tilde{E}_{i}(\mu)=E_{i}(\mu_{i})$,
$\tilde{L}_{i}(\mu)=L_{i}(\mu_{i})$, $\tilde{H}_{i}(\mu)=H_{i}(\mu_{i})$ for all
$\mu\in\mathcal{M}_{S}$ and $\mu_{i}=\psi_{i}^{-1}(\mu)\in\mathcal{M}_{i}$, $i\in\mathcal{N}$.
The initial state $\tilde{x}_{i0}=x_{i0}$, $i\in\mathcal{N}$.
The uncertainties $\tilde{\xi}_{i}(t)$, $\tilde{r}_{i}(t)$, $\tilde{\xi}^{u}_{i}(t)$, $i\in\mathcal{N}$,
satisfy the following constraints, respectively.

\begin{defn}
  \label{new:defn:un:local}
  A locally square integrable signal
  $[\tilde{\xi}_{1}^{T}(t),\ldots,\tilde{\xi}_{N}^{T}(t)]^{T}$ represents an
  admissible local uncertainty input for the large-scale system~\eqref{sys:large}
  if, given any locally square integrable signals
  $[\tilde{u}_{1}^{T}(t),\ldots,\tilde{u}_{N}^{T}(t)]^{T}$,
  $[\tilde{\xi}^{uT}_{1}(t),\ldots,\tilde{\xi}^{uT}_{N}(t)]^{T}$,
  $[\tilde{r}_{1}^{T}(t),\ldots,\tilde{r}_{N}^{T}(t)]^{T}$, there exists a time sequence $\set{t_{l}}_{l=1}^{\infty}$,
   $t_{l}\to\infty$, such that for all $l$ and for all $i\in\mathcal{N}$,
  \begin{align}
    \E\left(\int_{0}^{t_{l}}\left[\norm{\tilde{\zeta}_{i}(t)}^{2}
     -\norm{\tilde{\xi}_{i}(t)}^{2}\right] dt\right)
    \ge -\tilde{x}_{i0}^{T}\bar{S}_{i}\tilde{x}_{i0}.  \label{new:un:local}
  \end{align}
  The set of all such admissible
  local uncertainty inputs is denoted by $\tilde\Xi$.
  \end{defn}

\begin{defn}
  \label{new:defn:un:ic}
  A locally square integrable signal  $[\tilde{r}_{1}^{T}(t),\ldots,\tilde{r}_{N}^{T}(t)]^{T}$ represents an admissible
  interconnection input for the large-scale system~\eqref{sys:large}
  if, given any locally square integrable signals $[\tilde{u}_{1}^{T}(t),\ldots,\tilde{u}_{N}^{T}(t)]^{T}$,
  $[\tilde{\xi}^{uT}_{1}(t),\ldots,\tilde{\xi}^{uT}_{N}(t)]^{T}$,
  $[\tilde{\xi}_{1}^{T}(t),\ldots,\tilde{\xi}_{N}^{T}(t)]^{T}$, there exists a time sequence $\set{t_{l}}_{l=1}^{\infty}$,
   $t_{l}\to\infty$, such that for all $l$ and for all $i\in\mathcal{N}$,
   \begin{align}
    &\E\left(\int_{0}^{t_{l}}\left[\sum\limits_{j=1,j\ne i}^{N} \norm{\tilde{\zeta}_{j}(t)}^{2} -
     \norm{\tilde{r}_{i}(t)}^{2} \right]
     dt\right)\ge -\tilde{x}_{i0}^{T}\tilde{S}_{i}\tilde{x}_{i0}. \label{new:un:ic}
  \end{align}
The set of all such admissible interconnection inputs is denoted by $\tilde\Pi$.
\end{defn}

\begin{defn}[\cite{XUP09:tac}]
  \label{dfn:1}
  Suppose $\beta^{u}_{i}(\mu)\in\mathbb{R}^{+}$,
  $i\in\mathcal{N}$, $\mu\in\mathcal{M}_{S}$. A locally square integrable
  signal $[\tilde{\xi}^{uT}_{1}(t),\ldots,\tilde{\xi}^{uT}_{N}(t)]^{T}$ represents an admissible input uncertainty for the large-scale system~\eqref{sys:large}
  if,  for all locally square integrable signals $[\tilde{u}_{1}^{T}(t),\ldots,\tilde{u}_{N}^{T}(t)]^{T}$,
   $[\tilde{\xi}_{1}^{T}(t),\ldots,\tilde{\xi}_{N}^{T}(t)]^{T}$,
   $[\tilde{r}_{1}^{T}(t),\ldots,\tilde{r}_{N}^{T}(t)]^{T}$ and for all $i\in\mathcal{N}$,
  \begin{align}
    \E\left(\beta^{u}_{i}(\eta(t))\norm{\tilde{x}_{i}(t)}^{2}
        -\norm{\tilde{\xi}^{u}_{i}(t)}^{2}\right)\ge 0 . \label{un:input}
  \end{align}
 The set of all such admissible input uncertainties is denoted by $\tilde\Xi^{u}$.
\end{defn}

We assume that the same sequences $\set{t_{l}}_{l=1}^{\infty}$ are chosen in
Definitions~\ref{new:defn:un:local},~\ref{new:defn:un:ic} whenever they correspond to the same
signals $[\tilde{\xi}_{1}^{T}(t),\ldots,\tilde{\xi}_{N}^{T}(t)]^{T}$,
$[\tilde{\xi}^{uT}_{1}(t),\ldots,\tilde{\xi}^{uT}_{N}(t)]^{T}$,
$[\tilde{r}_{1}^{T}(t),\ldots,\tilde{r}_{N}^{T}(t)]^{T}$,
$[\tilde{u}_{1}^{T}(t),\ldots,\tilde{u}_{N}^{T}(t)]^{T}$.  Furthermore, one
can verify that the system~\eqref{sys:large} has the same system matrices as
the system~\eqref{sys}, i.e., $\tilde{A}_{i}(\cdot)=A_{i}(\cdot)$,
$\tilde{B}_{i}(\cdot)=B_{i}(\cdot)$, $\tilde{E}_{i}(\cdot)=E_{i}(\cdot)$,
$\tilde{L}_{i}(\cdot)=L_{i}(\cdot)$, $\tilde{H}_{i}(\cdot)=H_{i}(\cdot)$ at
any time $t$. Using this fact, we will show that $\tilde\Xi=\Xi$,
$\tilde\Pi=\Pi$.

For convenience, let $\mathcal{L}^{m}(t)$ denote the set of all locally square integrable signals of dimension  $m=\sum_{i=1}^{N}m_{i}$, and let $\mathcal{L}^{s}(t)$ denote the set of all locally square integrable signals of dimension  $s=\sum_{i=1}^{N}s_{i}$. Given $[\tilde{\xi}_{1}^{T}(t),\ldots,\tilde{\xi}_{N}^{T}(t)]^{T}\in\tilde\Xi$.
By Definition~\ref{new:defn:un:local}, the inequality~\eqref{new:un:local} holds for any signals
$[\tilde{u}_{1}^{T}(t),\ldots,\tilde{u}_{N}^{T}(t)]^{T}\in\mathcal{L}^{m}(t)$,
$[\tilde{\xi}^{uT}_{1}(t),\ldots,\tilde{\xi}^{uT}_{N}(t)]^{T}\in\mathcal{L}^{m}(t)$,
$[\tilde{r}_{1}^{T}(t),\ldots,\tilde{r}_{N}^{T}(t)]^{T}\in\mathcal{L}^{s}(t)$.
This implies that the inequality~\eqref{new:un:local} holds for any
$[\tilde{u}_{1}^{T}(t),\ldots,\tilde{u}_{N}^{T}(t)]^{T}\in\mathcal{L}^{m}(t)$,
$[\tilde{r}_{1}^{T}(t),\ldots,\tilde{r}_{N}^{T}(t)]^{T}\in\mathcal{L}^{s}(t)$ and
$[\tilde{\xi}^{uT}_{1}(t),\ldots,\tilde{\xi}^{uT}_{N}(t)]^{T}\equiv 0$.
This is indeed the case defined by Definition~\ref{defn:un:local}.
Thus we have $[\tilde{\xi}_{1}^{T}(t),\ldots,\tilde{\xi}_{N}^{T}(t)]^{T}\in\Xi$, i.e., $\tilde\Xi\subset\Xi$.

To show $\Xi\subset\tilde\Xi$, suppose
$[\xi_{1}^{T}(t),\ldots,\xi_{N}^{T}(t)]^{T}\in\Xi$. Then we shall prove
$[\xi_{1}^{T}(t),\ldots,\xi_{N}^{T}(t)]^{T}\in\tilde\Xi$. By
Definition~\ref{new:defn:un:local}, we need to prove that the
inequality~\eqref{new:un:local} holds when we apply this signal
$[\xi_{1}^{T}(t),\ldots,\xi_{N}^{T}(t)]^{T}$ and any other signals
$[\tilde{u}_{1}^{T}(t),\ldots,\tilde{u}_{N}^{T}(t)]^{T}\in\mathcal{L}^{m}(t)$,
$[\tilde{\xi}^{uT}_{1}(t),\ldots,\tilde{\xi}^{uT}_{N}(t)]^{T}\in\mathcal{L}^{m}(t)$,
$[\tilde{r}_{1}^{T}(t),\ldots,\tilde{r}_{N}^{T}(t)]^{T}\in\mathcal{L}^{s}(t)$
to the large-scale system~\eqref{sys:large}. Note that the two inputs
$[\tilde{u}_{1}^{T}(t),\ldots,\tilde{u}_{N}^{T}(t)]^{T}$,
$[\tilde{\xi}^{uT}_{1}(t),\ldots,\tilde{\xi}^{uT}_{N}(t)]^{T}$ in the
large-scale system~\eqref{sys:large} can be considered as an equivalent input
$[\hat{u}_{1}^{T}(t),\ldots,\hat{u}_{N}^{T}(t)]^{T}$.  For any
$[\tilde{u}_{1}^{T}(t),\ldots,\tilde{u}_{N}^{T}(t)]^{T}\in\mathcal{L}^{m}(t)$,
$[\tilde{\xi}^{uT}_{1}(t),\ldots,\tilde{\xi}^{uT}_{N}(t)]^{T}\in\mathcal{L}^{m}(t)$,
we have
$[\hat{u}_{1}^{T}(t),\ldots,\hat{u}_{N}^{T}(t)]^{T}\in\mathcal{L}^{m}(t)$.
Thus, it suffices to prove that the inequality~\eqref{new:un:local} holds when
we apply $[\xi_{1}^{T}(t),\ldots,\xi_{N}^{T}(t)]^{T}$ and any other signals
$[\hat{u}_{1}^{T}(t),\ldots,\hat{u}_{N}^{T}(t)]^{T}\in\mathcal{L}^{m}(t)$,
$[\tilde{r}_{1}^{T}(t),\ldots,\tilde{r}_{N}^{T}(t)]^{T}\in\mathcal{L}^{s}(t)$.
But this follows directly from Definition~\ref{defn:un:local} and the fact
that $[\xi_{1}^{T}(t),\ldots,\xi_{N}^{T}(t)]^{T}\in\Xi$.  Hence
$[\xi_{1}^{T}(t),\ldots,\xi_{N}^{T}(t)]^{T}\in\tilde\Xi$ and
$\Xi\subset\tilde\Xi$. Therefore $\tilde\Xi=\Xi$. In a similar way, we can
prove that $\tilde\Pi=\Pi$.

We also mention that the values of $\beta^{u}_{i}(\mu)$, $\mu\in\mathcal{M}_{S}$, $i\in\mathcal{N}$, in Definition~\ref{dfn:1} can either be given appropriately in advance,
or be solved from numerical computation as illustrated in Theorem~\ref{thm3}.

Associated with the large-scale system~\eqref{sys:large} is the quadratic cost functional
as follows~\cite{XUP09:tac}:
  \begin{align}
    J\triangleq\E&\left(\int_{0}^{\infty}\sum\limits_{i=1}^{N}[\tilde{x}_{i}^{T}(t)\tilde{R}_{i}(\eta(t))\tilde{x}_{i}(t)\vphantom{\int_{0}^{\infty}\sum\limits_{i=1}^{N}}\right.\notag\\
    &\left.\vphantom{\int_{0}^{\infty}\sum\limits_{i=1}^{N}}+\tilde{u}_{i}^{T}(t)\tilde{G}_{i}(\eta(t))\tilde{u}_{i}(t)]dt\right), \label{performance}
  \end{align}
  where $\tilde{R}_{i}(\mu)\in\mathbb{S}^{+}$, $\tilde{G}_{i}(\mu)\in\mathbb{S}^{+}$, $\mu\in\mathcal{M}_{S}$, $i\in\mathcal{N}$, are given weighting matrices.

For the large-scale system~\eqref{sys:large}
with the uncertainty constraints~\eqref{new:un:local}, \eqref{new:un:ic}, \eqref{un:input},
global mode dependent stabilizing controllers can be
designed using the technique developed in~\cite{UP05:ijc}.
Furthermore, applying these controllers to the large-scale system~\eqref{sys:large}
will yield a cost upper bound, i.e., sup$_{\tilde\Xi,\tilde\Pi,\tilde\Xi^{u}}J<c_{1}$, $c_{1}\in\mathbb{R}^{+}$.
This result is stated in the following theorem.

\begin{thm}[\cite{XUP09:tac}]
  \label{thm1}
   If there exist matrices $X_{i}(\mu)\in\mathbb{S}^{+}$ and
   scalars $\tau_{i}\in\mathbb{R}^{+}$,
      $\theta_{i}\in\mathbb{R}^{+}$, $\tau^{u}_{i}\in\mathbb{R}^{+}$,
       $\mu\in\mathcal{M}_{S}$, $i\in\mathcal{N}$, such that
            \begin{align}\label{eq:are}
         &\tilde{A}_{i}^{T}(\mu)X_{i}(\mu) + X_{i}(\mu)\tilde{A}_{i}(\mu)
        + \sum\limits_{\nu=1}^{M}q_{\mu\nu}X_{i}(\nu)
        +\tilde{R}_{i}(\mu)\notag\\
          &+ X_{i}(\mu)\left(\bar{B}_{2i}(\mu)\bar{B}_{2i}^{T}(\mu) -\tilde{B}_{i}(\mu)\tilde{G}_{i}^{-1}(\mu)\tilde{B}_{i}^{T}(\mu)\right)X_{i}(\mu)\notag\\
    &+ \tau^{u}_{i}\beta ^{u}_{i}(\mu)I+(\tau_{i}+\bar{\theta}_{i})\tilde{H}_{i}^{T}(\mu)\tilde{H}_{i}(\mu)  < 0,
      \end{align}
      where
        $\bar{B}_{2i}(\mu) =
          \begin{bmatrix}
          (\tau^{u}_{i})^{-1/2}\tilde{B}_{i}(\mu)\; \tau_{i}^{-1/2}\tilde{E}_{i}(\mu)\; \theta_{i}^{-1/2}\tilde{L}_{i}(\mu)
          \end{bmatrix}$
       and $\bar{\theta}_{i}=\sum_{j=1,j\ne i}^{N}\theta_{j}$,
       then the global mode dependent controllers given by
      \begin{equation}\label{ctr:large:design}
      \left\{\begin{aligned}
       \tilde{u}_{i}(t) &= \tilde{K}_{i}(\eta(t))\tilde{x}_{i}(t), \\
        \tilde{K}_{i}(\mu) &= -\tilde{G}_{i}^{-1}(\mu)\tilde{B}_{i}^{T}(\mu)X_{i}(\mu),
      \end{aligned}\right.
      \end{equation}
      $\mu\in\mathcal{M}_{S}$, $i\in\mathcal{N}$, robustly stabilize the uncertain large-scale
      system~\eqref{sys:large} with the uncertainty constraints~\eqref{new:un:local}, \eqref{new:un:ic}, \eqref{un:input}, and achieve a bounded system cost $ J \le
      \sum_{i=1}^{N}\tilde{x}_{i0}^{T}\left[\sum_{\mu=1}^{M}\pi_{\mu}X_{i}(\mu)+\tau_{i}\bar{S}_{i}
        +\theta_{i}\tilde{S}_{i}\right]\tilde{x}_{i0}$.
\end{thm}

After obtaining the global mode dependent stabilizing controllers~\eqref{ctr:large:design}
for the large-scale system~\eqref{sys:large}, the next step is to derive neighboring mode
dependent stabilizing controllers for the large-scale system~\eqref{sys}.
The following result is an extension of Theorem~1 in~\cite{XUP09:tac} to
the neighboring mode dependent control case.
The proof is similar to that of Theorem~1 in~\cite{XUP09:tac} and hence is omitted.

\begin{thm}
  \label{thm2}
  Given the global mode dependent controllers~\eqref{ctr:large:design}
  which stabilize the large-scale system~\eqref{sys:large} with the uncertainty
  constraints~\eqref{new:un:local},~\eqref{new:un:ic},~\eqref{un:input}.
  If the gains $K_{i}(\cdot)$ in the controllers~\eqref{ctr} are chosen to satisfy
  \begin{align}
    \label{eq:thm2}
    \norm{K_{i}(\sigma_{i})-\tilde{K}_{i}(\mu)}^{2} &\le \beta^{u}_{i}(\mu)
  \end{align}
  for all $\mu\in\mathcal{M}_{S}$,  $\sigma_{i}=\phi_{i}(\mu)\in\mathcal{M}_{S i}$, $i\in\mathcal{N}$,
  then the neighboring mode dependent controllers~\eqref{ctr}
  stabilize the large-scale system~\eqref{sys}
  with the uncertainty constraints~\eqref{un:local},~\eqref{un:ic}.
\end{thm}

In the following remark, we use an example to illustrate the fact that Theorem~\ref{thm2} is less conservative than Theorem~1 in~\cite{XUP09:tac}.

\begin{rem}
Consider the Markovian jump large-scale system in Example~\ref{exmp:1}.
Three control gains need to be scheduled for the first local controller
if using the neighboring mode dependent control approach, while two control gains are needed
if using the local mode dependent control approach.
We denote the three neighboring mode dependent control gains
as $K_{1}(\sigma_{1}), \sigma_{1}\in\mathcal{M}_{S 1}=\{1, 2, 3\}$
and the two local mode dependent control gains as
$\mathcal{K}_{1}(\mu_{1}), \mu_{1}\in\mathcal{M}_{1}=\{1, 2\}$.
For comparison, given the global mode dependent control gains $\tilde{K}_{1}(\mu)$
and the scalars $\beta^{u}_{1}(\mu), \mu=1,2,3$, the constraints imposed
on $\mathcal{K}_{1}(1)$, $K_{1}(1)$  are specified as follows
based on Theorem~1 in~\cite{XUP09:tac} and our Theorem~\ref{thm2}, respectively:
  \begin{align}
    \left\{
    \begin{aligned}
    \norm{\mathcal{K}_{1}(1)-\tilde{K}_{1}(1)}^{2} &\le \beta^{u}_{1}(1),\\
    \norm{\mathcal{K}_{1}(1)-\tilde{K}_{1}(2)}^{2} &\le \beta^{u}_{1}(2),\\
    \norm{\mathcal{K}_{1}(1)-\tilde{K}_{1}(3)}^{2} &\le \beta^{u}_{1}(3),
    \end{aligned}\right.\\
    \left\{
    \begin{aligned}
    \norm{K_{1}(1)-\tilde{K}_{1}(1)}^{2} &\le \beta^{u}_{1}(1),\\
    \norm{K_{1}(1)-\tilde{K}_{1}(2)}^{2} &\le \beta^{u}_{1}(2).
    \end{aligned}\right.
  \end{align}
  These inequalities are illustrated in Fig.~\ref{fig1} where each circle
  denotes a Euclidean ball.  $\tilde{K}_{1}(\mu)$ is the center and
  $\sqrt{\beta^{u}_{1}(\mu)}$ the radius of the ball for $\mu=1,2,3$.  As
  shown in Fig.~\ref{fig1}, the set where $\mathcal{K}_{1}(1)$ takes values is
  only a subset of the set where $K_{1}(1)$ takes values. Hence the proposed
  framework provides greater flexibility in choosing control gains.
  Potentially, this will allow one to achieve better system performance than
  obtained using local mode dependent controllers.  We also mention that if
  the Euclidean ball centered at $\tilde{K}_{1}(3)$ does not intersect the
  Euclidean ball centered at $\tilde{K}_{1}(1)$ (or $\tilde{K}_{1}(2)$), then no
  local mode dependent controllers exist. However, the existence of the
  neighboring mode dependent controllers is not affected.  Therefore our
  technique potentially produces less conservative results than that
  in~\cite{XUP09:tac}.
\end{rem}

\begin{figure}[htbp]
\begin{center}
\includegraphics[width=6cm]{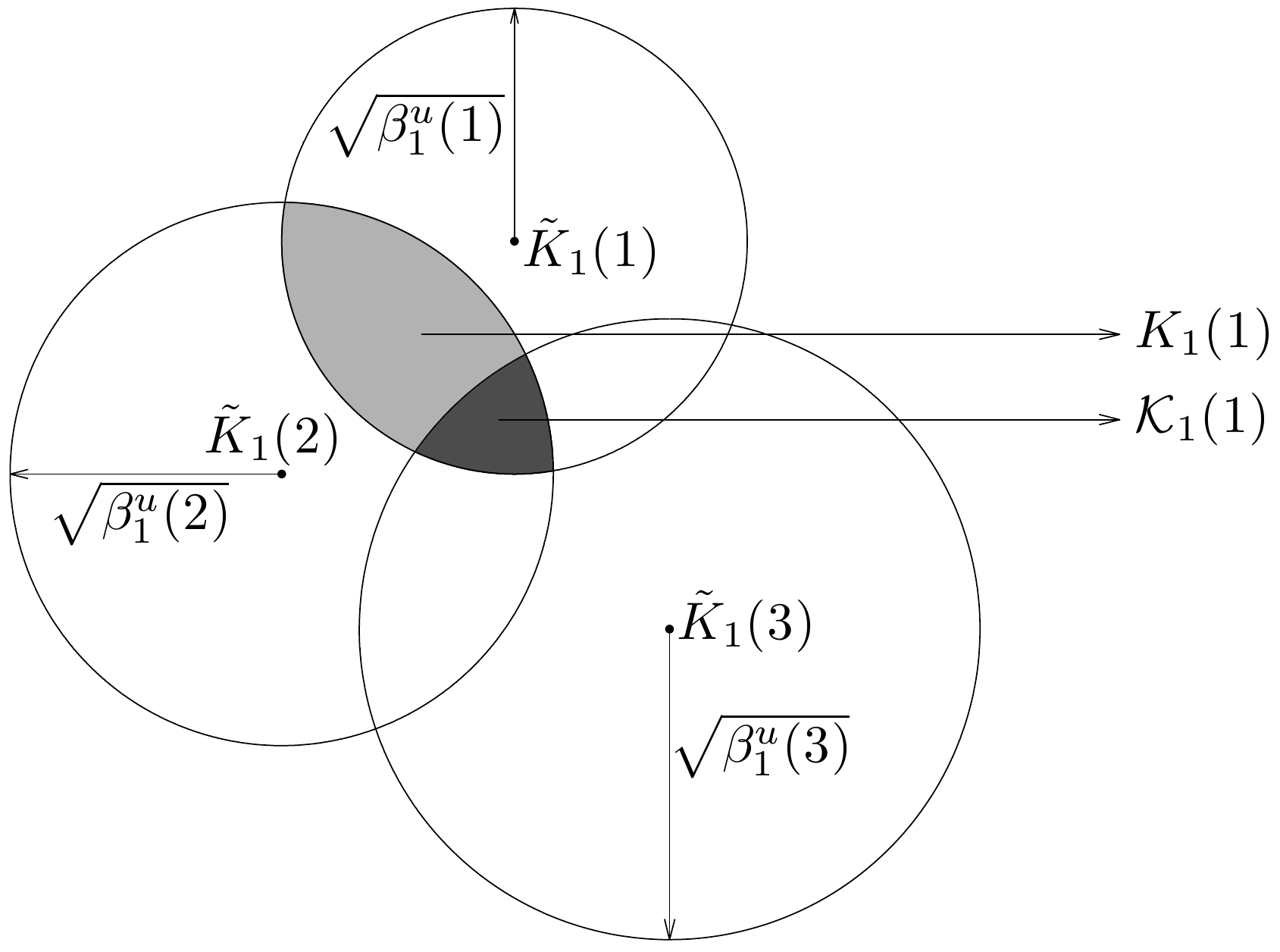}
\caption{Illustration of the constraints.}
\label{fig1}
\end{center}
\end{figure}

Next, the conditions in Theorem~\ref{thm1} and Theorem~\ref{thm2} are combined
and recast as a rank constrained LMI problem.  Although rank constrained LMI
problems are non-convex in general, numerical methods such as the LMIRank
toolbox~\cite{lmirank} often yield good results in solving these problems.

\begin{thm}
  \label{thm3}
  Suppose there exist matrices $X_{i}(\mu)\in\mathbb{S}^{+}$, $Y_{i}(\mu)\in\mathbb{S}^{+}$,
 $K_{i}(\sigma_{i})\in\mathbb{R}^{m_{i} \times n_{i}}$ and scalars
 $\bar{\beta}_{i}(\mu)\in\mathbb{R}^{+}$, $\tilde{\beta}_{i}(\mu)\in\mathbb{R}^{+}$,
$\tilde{\tau}^{u}_{i}\in\mathbb{R}^{+}$, $\tilde{\tau}_{i}\in\mathbb{R}^{+}$,
  $\tilde{\theta}_{i}\in\mathbb{R}^{+}$, $\mu\in\mathcal{M}_{S}$, $i\in\mathcal{N}$,
  such that the following inequalities hold:
  \begin{align}
  \begin{bmatrix}
    \mathcal{G}_{i11}(\mu) & \mathcal{G}_{i12}(\mu) & \mathcal{G}_{i13}(\mu)\\
    \mathcal{G}_{i12}^{T}(\mu) & \mathcal{G}_{i22}(\mu) &0\\
    \mathcal{G}_{i13}^{T}(\mu) &0 &\mathcal{G}_{i33}(\mu)
  \end{bmatrix} &< 0, \label{eq:lmi} \\
  \begin{bmatrix}
    -\tilde{\tau}^{u}_{i}I & \Upsilon_{i}^{T}(\mu) \\
    \Upsilon_{i}(\mu) & -\tilde{\beta}_{i}(\mu)I
  \end{bmatrix} &\le 0, \label{eq:theorem2lmi}\\
    \rank\left(
  \begin{bmatrix}
    \bar{\beta}_{i}(\mu) & 1 \\
    1 & \tilde{\beta}_{i}(\mu)
  \end{bmatrix}\right)&\le 1, \label{beta_inverse} \\
  \rank\left(
  \begin{bmatrix}
    Y_{i}(\mu) & I \\
    I & X_{i}(\mu)
  \end{bmatrix}\right)&\le n_{i},  \label{X_inverse}
\end{align}
where
\begin{align*}
  \mathcal{G}_{i11}(\mu) &= Y_{i}(\mu)\tilde{A}_{i}^{T}(\mu)
      + \tilde{A}_{i}(\mu)Y_{i}(\mu) + q_{\mu\mu}Y_{i}(\mu)\\
    &\quad -\tilde{B}_{i}(\mu)\tilde{G}_{i}^{-1}(\mu)\tilde{B}_{i}^{T}(\mu)+\tilde{\tau}^{u}_{i}\tilde{B}_{i}(\mu)\tilde{B}_{i}^{T}(\mu)\\
    &\quad+\tilde{\tau}_{i}\tilde{E}_{i}(\mu)\tilde{E}_{i}^{T}(\mu) + \tilde{\theta}_{i}\tilde{L}_{i}(\mu)\tilde{L}_{i}^{T}(\mu), \\
  \mathcal{G}_{i12}(\mu) &=
    Y_{i}(\mu)[\sqrt{q_{\mu 1}}I\; \cdots\; \sqrt{q_{\mu (\mu-1)}}I \\
    &\quad \sqrt{q_{\mu (\mu+1)}}I\; \cdots\; \sqrt{q_{\mu M}}I],   \\
      \mathcal{G}_{i13}(\mu) &=
    Y_{i}(\mu)\left[I\;I\;\tilde{H}_{i}^{T}(\mu)\;\cdots\;\tilde{H}_{i}^{T}(\mu)\right],   \\
  \mathcal{G}_{i22}(\mu)  &= -\diag[Y_{i}(1),\cdots,Y_{i}(\mu-1),\\
  &\quad Y_{i}(\mu+1),\cdots,Y_{i}(M)],\\
      \mathcal{G}_{i33}(\mu)  &= -\diag[\tilde{R}_{i}^{-1}(\mu),\bar{\beta}_{i}(\mu)I,\tilde{\tau}_{i}I,\tilde{\theta}_{1}I,\cdots,\tilde{\theta}_{i-1}I,\\
   &\quad \tilde{\theta}_{i+1}I,\cdots,\tilde{\theta}_{N}I],\\
  \Upsilon_{i}(\mu) &=K_{i}(\phi_{i}(\mu))+\tilde{G}_{i}^{-1}(\mu)\tilde{B}_{i}^{T}(\mu)X_{i}(\mu).
\end{align*}
Then a stabilizing controller~\eqref{ctr} is given by:
$u_{i}(t)= K_{i}(\sigma_{i})x_{i}(t)$, for $\aleph_{i}(t)=\sigma_{i}\in\mathcal{M}_{S i}$, $i\in\mathcal{N}$.
\end{thm}

\begin{proof}
From $X_{i}(\mu)\in\mathbb{S}^{+}$, $Y_{i}(\mu)\in\mathbb{S}^{+}$ and~\eqref{X_inverse},
we have $Y_{i}(\mu)=(X_{i}(\mu))^{-1}$.
Similarly, $\tilde{\beta}_{i}(\mu)=(\bar{\beta}_{i}(\mu))^{-1}$.
On the other hand, if~\eqref{eq:lmi} is satisfied, by setting  $\tau^{u}_{i}=(\tilde{\tau}^{u}_{i})^{-1}$, $\tau_{i}=(\tilde{\tau}_{i})^{-1}$,   $\theta_{i}=(\tilde{\theta}_{i})^{-1}$,
$\beta_{i}^{u}(\mu)=(\bar{\beta}_{i}(\mu))^{-1}\tilde{\tau}^{u}_{i}=\tilde{\beta}_{i}(\mu)\tilde{\tau}^{u}_{i}$, and applying the Schur complement equivalence, the inequality~\eqref{eq:are} is satisfied.
Then, by Theorem~\ref{thm1}, the global mode dependent controllers~\eqref{ctr:large:design}
can be designed to stabilize the large-scale system~\eqref{sys:large}
with the uncertainty constraints~\eqref{new:un:local}, \eqref{new:un:ic}, \eqref{un:input}.

Also, the LMI~\eqref{eq:theorem2lmi} and the
equation~\eqref{ctr:large:design} imply that
$\norm{\tilde{K}_{i}(\mu)-K_{i}(\sigma_{i})}^{2}\le\beta_{i}^{u}(\mu)$ for
all $\mu\in\mathcal{M}_{S}$, $\sigma_{i}=\phi_{i}(\mu)$, $i\in\mathcal{N}$.
That is, the inequality~\eqref{eq:thm2} holds.  Then, by Theorem~\ref{thm2},
the constructed controllers~\eqref{ctr}
stabilize the large-scale system~\eqref{sys} with
the uncertainty constraints~\eqref{un:local},~\eqref{un:ic}.
\end{proof}

\begin{rem}
  In~\cite{XUP09:tac}, a control gain form has been proposed for the design of
  local mode dependent controllers. That is, each local mode dependent control
  gain is chosen to be a weighted average of the related global mode dependent
  control gains.  This particular gain form is then incorporated into the
  coupled LMIs from which the local mode dependent control gains are computed;
  see Theorem~3 and Theorem~4 in~\cite{XUP09:tac} for details.  Unfortunately,
  choosing such a gain form is not helpful in terms of an improvement in
  system performance, and sometimes may even result in infeasibility of the
  corresponding LMIs.  A demonstration of this fact is given in
  Section~\ref{sec:ie}.  Indeed, such a gain form imposes an additional
  constraint and hence is not used in this paper.
\end{rem}

\section{Numerical Example} \label{sec:ie}
Consider the Markovian jump large-scale system given in~\cite{XUP09:tac}.  The
 mode information is
$\mathcal{M}_{V}=\set{[1,1,1]^{T},[1,2,2]^{T},[2,1,2]^{T},[2,2,1]^{T}}$.
The initial distribution of $\eta(t)$ is assumed to be the same as its
stationary distribution $\pi_{\infty}=[\pi_{\infty 1},\ldots,\pi_{\infty
  M}]^{T}$, which can be computed from the infinitesimal generator matrix
$\mathbf{Q}$.  Given a neighboring mode information pattern $\mathcal{C}$, our
objective is to find the corresponding neighboring mode dependent stabilizing
controllers for this large-scale system. An upper bound on the quadratic
cost~\eqref{performance} is also evaluated for the resulting closed-loop
large-scale system. The main software we use is the LMIRank
toolbox~\cite{lmirank}. The procedure is summarized as follows:
\begin{enumerate}
\item Solve the optimization problem
  \begin{align*}
 &\min\gamma \quad \text{subject to} \\
 &\sum\limits_{i=1}^{N}x_{i0}^{T}\left[\sum\limits_{\mu=1}^{M}
  \pi_{\infty\mu}X_{i}(\mu)+\tau_{i}\bar{S}_{i}+\theta_{i}\tilde{S}_{i}
\right]x_{i0}<\gamma,\\
&\text{and~\eqref{eq:lmi},~\eqref{eq:theorem2lmi},~\eqref{beta_inverse},~\eqref{X_inverse}}.
 \end{align*}
  If an optimal value $\gamma$ is found, feasible neighboring mode dependent
control gains~\eqref{ctr} are obtained.
\item Apply the obtained controllers to the large-scale system~\eqref{sys} and
  compute the cost upper bound for the resulting closed-loop large-scale
  system.  The method for computing this upper bound is taken
  from~\cite{UP05:ijc}. It involves solving a worst-case performance analysis
  problem.
\end{enumerate}

Five cases are considered, i.e.,
\begin{align*}
\mathcal{C}_{1} &= \begin{bmatrix} 1 & 0 & 0 \\ 0 & 1 & 0 \\ 0 & 0 & 1 \\ \end{bmatrix},
\mathcal{C}_{2} = \begin{bmatrix} 1 & 1 & 0 \\ 0 & 1 & 0 \\ 0 & 0 & 1 \\ \end{bmatrix},
\mathcal{C}_{3} = \begin{bmatrix} 1 & 1 & 0 \\ 0 & 1 & 1 \\ 0 & 0 & 1 \\ \end{bmatrix},\\
\mathcal{C}_{4} &= \begin{bmatrix} 1 & 1 & 0 \\ 0 & 1 & 1 \\ 0 & 1 & 1 \\ \end{bmatrix},
\mathcal{C}_{5} = \begin{bmatrix} 1 & 1 & 1 \\ 1 & 1 & 1 \\ 1 & 1 & 1 \\ \end{bmatrix}.
\end{align*}
It can be seen that each neighboring mode information pattern contains more
mode information than the preceding one. $\mathcal{C}_{1}$ corresponds to the
local mode dependent control case, while $\mathcal{C}_{5}$ corresponds to the
global mode dependent control case.  By using the preceding procedure,
neighboring mode dependent stabilizing controllers are found for each of these
cases.  Furthermore, if we apply the obtained controllers to the large-scale
system, the cost upper bounds for the resulting closed-loop large-scale
systems are shown in~Fig.~\ref{fig2}.

\begin{figure}[htbp]
\begin{center}
\includegraphics[width=7cm]{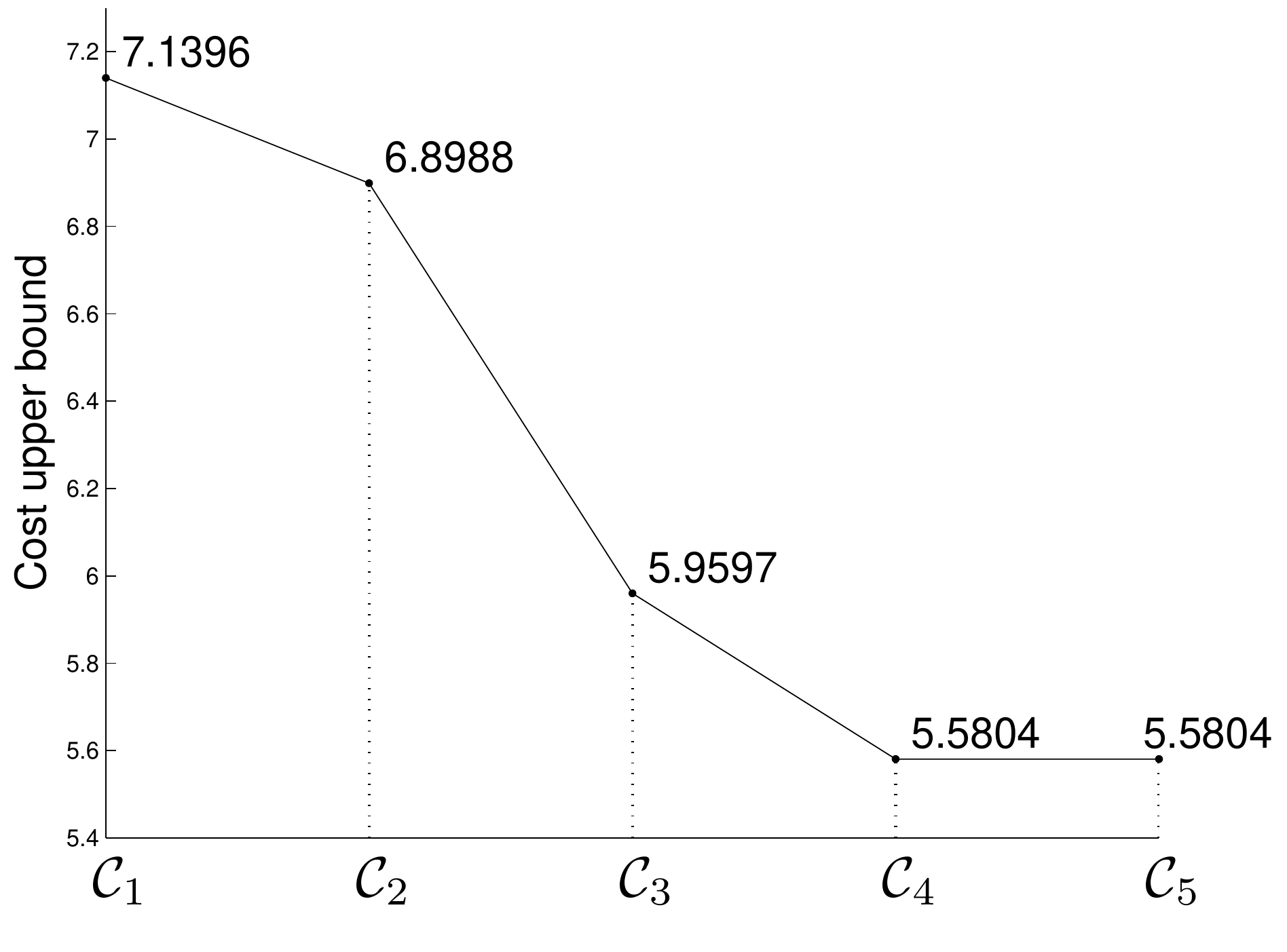}
\caption{Cost upper bounds for the closed-loop systems.}
\label{fig2}
\end{center}
\end{figure}

Note that the cost upper bound found here in the local mode dependent control
case is different from (in fact, less than) that in~\cite{XUP09:tac}.  This is
because the gain form proposed by Theorem~3 in~\cite{XUP09:tac} is not used in
our computation.  One may also notice that the cost upper bound found in the
case of $\mathcal{C}_{4}$ is the same as the one in the case of
$\mathcal{C}_{5}$.  We now explain why this happens. In the case of
$\mathcal{C}_{4}$, each local controller obtains two subsystem modes directly. In fact,  the third subsystem mode can be 
derived from these two modes based on possible mode combinations in $\mathcal{M}_{V}$. Hence
$\mathcal{C}_{4}$ and $\mathcal{C}_{5}$ are equivalent, in the sense that they yield the same performance.
This example demonstrates that the system achieves better (or at
least equal) performance if more
information about the subsystem modes is available to the local
controllers. It also shows that sometimes complete information about the
global mode of the large-scale system may be redundant.

\section{Conclusions}
\label{sec:conclusions}
This paper has presented a decentralized control scheme for uncertain
Markovian jump large-scale systems. The proposed controllers use local subsystem states
and neighboring mode information to generate local control inputs. A computational algorithm
involving rank constrained LMIs has been developed for the design of such controllers.
The developed theory is illustrated by a numerical example.


\begin{thebibliography}{10}

\bibitem{S06:tac}
C.~E. de~Souza.
\newblock Robust stability and stabilization of uncertain discrete-time
  {M}arkovian jump linear systems.
\newblock {\em IEEE Transactions on Automatic Control}, 51(5):836--841, 2006.

\bibitem{DY08:auto}
J.~Dong and G.~H. Yang.
\newblock Robust ${H}_{2}$ control of continuous-time {M}arkov jump linear
  systems.
\newblock {\em Automatica}, 44(5):1431--1436, 2008.

\bibitem{FGS09:auto}
Z.~Fei, H.~Gao, and P.~Shi.
\newblock New results on stabilization of {M}arkovian jump systems with time
  delay.
\newblock {\em Automatica}, 45(10):2300--2306, 2009.

\bibitem{FLS10:tac}
J.~E. Feng, J.~Lam, and Z.~Shu.
\newblock Stabilization of {M}arkovian systems via probability rate synthesis
  and output feedback.
\newblock {\em IEEE Transactions on Automatic Control}, 55(3):773--777, 2010.

\bibitem{LU07:tac}
L.~Li and V.~Ugrinovskii.
\newblock On necessary and sufficient conditions for ${H}_\infty$ output
  feedback control of {M}arkov jump linear systems.
\newblock {\em IEEE Transactions on Automatic Control}, 52(7):1287--1292, 2007.

\bibitem{LUO07:auto}
L.~Li, V.~Ugrinovskii, and R.~Orsi.
\newblock Decentralized robust control of uncertain {M}arkov jump parameter
  systems via output feedback.
\newblock {\em Automatica}, 43(11):1932--1944, 2007.

\bibitem{MY06:book}
X.~Mao and C.~Yuan.
\newblock {\em Stochastic differential equations with {M}arkovian switching}.
\newblock Imperial College Press, London, 2006.

\bibitem{lmirank}
R.~Orsi.
\newblock Lmirank: Software for rank constrained {{LMI}} problems.
\newblock \url{http://users.cecs.anu.edu.au/~robert/}.

\bibitem{PUS00:book}
I.~R. Petersen, V.~Ugrinovskii, and A.~V. Savkin.
\newblock {\em Robust Control Design Using ${H}_\infty$ Methods}.
\newblock Springer, London, 2000.

\bibitem{SBA991:tac}
P.~Shi, E.~K. Boukas, and R.~K. Agarwal.
\newblock Kalman filtering for continuous-time uncertain systems with
  {M}arkovian jumping parameters.
\newblock {\em IEEE Transactions on Automatic Control}, 44(8):1592--1597, 1999.

\bibitem{SLXZ07:auto}
M.~Sun, J.~Lam, S.~Xu, and Y.~Zou.
\newblock Robust exponential stabilization for {M}arkovian jump systems with
  mode-dependent input delay.
\newblock {\em Automatica}, 43(10):1799--1807, 2007.

\bibitem{UP05:ijc}
V.~Ugrinovskii and H.~R. Pota.
\newblock Decentralized control of power systems via robust control of
  uncertain {M}arkov jump parameter systems.
\newblock {\em International Journal of Control}, 78(9):662--677, 2005.

\bibitem{WSGW08:auto}
L.~Wu, P.~Shi, H.~Gao, and C.~Wang.
\newblock ${H}_\infty$ filtering for 2{D} {M}arkovian jump systems.
\newblock {\em Automatica}, 44(7):1849--1858, 2008.

\bibitem{XUP09:tac}
J.~Xiong, V.~Ugrinovskii, and I.~R. Petersen.
\newblock Local mode dependent decentralized stabilization of uncertain
  {M}arkovian jump large-scale systems.
\newblock {\em IEEE Transactions on Automatic Control}, 54(11):2632--2637,
  2009.

\bibitem{XUP10:book}
J.~Xiong, V.~Ugrinovskii, and I.~R. Petersen.
\newblock Decentralized output feedback guaranteed cost control of uncertain
  {M}arkovian jump large-scale systems: local mode dependent control approach.
\newblock In J.~Mohammadpour and K.~M. Grigoriadis, editors, {\em Efficient
  Modeling and Control of Large-Scale Systems}, pages 167--196. Springer, New
  York, 2010.

\bibitem{XC02:tac}
S.~Xu and T.~Chen.
\newblock Robust ${H}_\infty$ control for uncertain stochastic systems with
  state delay.
\newblock {\em IEEE Transactions on Automatic Control}, 47(12):2089--2094,
  2002.

\bibitem{ZHL03:scl}
L.~Zhang, B.~Huang, and J.~Lam.
\newblock ${H}_\infty$ model reduction of {M}arkovian jump linear systems.
\newblock {\em System $\&$ Control Letters}, 50(2):103--118, 2003.

\end{thebibliography}

\end{document}